\documentclass{article}
\pdfoutput=1
\usepackage{supertech}
\usepackage{xspace}
\usepackage{color}
\usepackage{amssymb}
\usepackage{amsmath}
\usepackage{subfig}
\usepackage{graphicx}
\usepackage{url}
\usepackage[linesnumbered,boxed,noend]{algorithm2e}

\newcommand{\cS}{{\cal S}}
\newcommand{\cF}{{\cal F}}

\newcommand{\RN}{Radio Network\xspace}

\newcommand{\diss}{Layer Dissemination\xspace}

\usepackage[normalem]{ulem}
\newcommand{\mig}[1]{\textcolor{blue}{#1}}

\renewcommand{\mig}[1]{#1}

\title{Ad-hoc Affectance-selective Families for\\ \diss}

\author{
Dariusz R. Kowalski~\thanks{Univ. of Liverpool, Computer Science Dept., Liverpool, UK.~\texttt{D.Kowalski@liverpool.ac.uk}}
\and
Harshita Kudaravalli~\thanks{Pace University, Computer Science Dept., New York, NY, USA.~\texttt{hk21040n@pace.edu}}
\and 
Miguel A. Mosteiro~\thanks{Pace University, Computer Science Dept., New York, NY, USA.~\texttt{mmosteiro@pace.edu}}
}

\date{}

\begin{document}

\maketitle

\begin{abstract}

Information dissemination protocols for ad-hoc wireless networks frequently use a minimal subset of the available 
communication links, defining a rooted ``broadcast'' tree. 
In this work, we focus on the core challenge of disseminating from one layer to the next one of such tree. 
We call this problem \emph{\diss}.
We study \diss under a generalized model of interference, called \emph{affectance}. 
The affectance model subsumes previous models, such as \RN and Signal to Inteference-plus-Noise Ratio. 
We present randomized and deterministic protocols for \diss. 
These protocols are based on a combinatorial object that we call \emph{Affectance-selective Families}.
\mig{
Our approach combines an engineering solution with theoretical guarantees. That is, we provide a method to characterize the network with a global measure of affectance based on measurements of interference in the specific deployment area. Then, our protocols distributedly produce an ad-hoc transmissions schedule for dissemination. In the randomized protocol only the network characterization is needed, whereas the deterministic protocol requires full knowledge of affectance. Our theoretical analysis provides guarantees on schedule length.
We also present simulations of a real network-deployment area contrasting the performance of our randomized protocol, which takes into account affectance, against previous work for interference models that ignore some physical constraints.
The striking improvement in performance shown by our simulations show the importance of utilizing a more physically-accurate model of interference that takes into account other effects beyond distance to transmitters.
}
 
\end{abstract}


\section{Introduction}
The problem of disseminating information in ad-hoc wireless communication networks (for instance, embedded in the Internet of Things) has been studied in theory and practice. 
%
To reduce traffic, dissemination protocols often use a minimal subset of the available communication links, call it $T$. Given that nodes communicate through radio broadcast, nodes may still receive through other links, but to provide performance guarantees only $T$ is assumed to be available, albeit taking into account the interference of the rest of the links.

When the dissemination task involves delivery to all nodes, $T$ defines a tree topology. (Since all nodes must be reachable but the set is minimal.) 
Either because there is a single source node (e.g.~\cite{KMR_fomc14,KowalskiMZarxiv15}), 
or because packets are first aggregated at a single node for later dissemination (e.g.~\cite{KhabbazianK11,manne2006optimal}), 
the problem reduces to disseminate from a root to all other nodes through a \emph{Broadcast Tree}. Moreover, as we observed in~\cite{KMR_fomc14,KowalskiMZarxiv15}, when packets are disseminated from layer to layer in a BFS fashion, the bottleneck for fast dissemination on broadcast trees occurs at layers with high interference. Indeed, we have shown in~\cite{KMR_fomc14,KowalskiMZarxiv15} that in the long run 
throughput is a function of maximum layer interference, and it is independent of interference in paths where packets can be easily pipelined. 
This phenomenon has also been observed in other works, such as in the following fragment in~\cite{GHK:mmbComplexity}.


\begin{quote}
\emph{In fact, if one has a fast way of transmitting one or more messages from one layer to the next, repeating this and using pipelining would yield a fast broadcast algorithm. Thus, the crux of the broadcast problem lies in how fast can this task be solved in bipartite graphs.}
\end{quote}

Thus, in this work, we focus on the core challenge of dissemination through one layer of a broadcast tree. We model such layer as a bipartite graph $G=(V,W,E)$ where $V$ (the \emph{transmitters}) and $W$ (the \emph{receivers}) are sets of nodes and $E$ is the set of links from $V$ to $W$. We study the \diss problem in $G$ assuming that initially all the transmitters have an identical piece of information, called \emph{message} or \emph{packet} indistinctively. To complete the task, all the nodes in $W$ have to receive the message. 

We do not assume any underlying communication infrastructure. That is, transmitters attempt to deliver the message by radio broadcast but, if two or more nodes transmit at the same time, mutual interference may prevent reception of the message. 
To take into account this phenomenon, we study \diss under a general model of interference called \emph{affectance}. As in~\cite{KMR_fomc14,KowalskiMZarxiv15} we parameterize affectance with a real value $0\leq a(u,(v,w))\leq 1$ that represents the affectance of each transmitter $u$ on each link $(v,w)$. 
An affectance model of interference from links on links was introduced by Kesselheim~\cite{Kaff} in the context of link scheduling.
\mig{
Affectance is a general model of interference in the sense that comprises other particular models studied before (cf.~\cite{KowalskiMZarxiv15}).
Moreover, previous models do not accurately represent the physical constraints in real-world deployments. 
For instance, 
in the \RN model~\cite{ChlamtacK87} interference from non-neighboring nodes is neglected, 
and Signal to Inteference-plus-Noise Ratio (SINR)~\cite{HWaff,SRSSINRdomSet} is a simplified model because other constraints, such as obstacles, are not taken into account.
}

\diss is closely related to the combinatorial problem of computing selective families. 
The notion of selective families was introduced in~\cite{CGGPRselFam} as a generalization of the dissemination problem in the \RN model to a combinatorial problem.
Later in~\cite{ClementiCMPSapprox01}, Clementi et al. showed how to compute selective families ad-hoc, that is, for a given input family. The results are applicable to dissemination under the \RN model of interference when the topology is known.

In this work, we follow-up on~\cite{CGGPRselFam} and~\cite{ClementiCMPSapprox01} introducing the concept of \emph{Affectance-selective Families}. That is, we generalize the dissemination problem in bipartite graphs also to a combinatorial problem, but taking into account the specific conditions to achieve a successful transmission under our generalized model of interference. Under certain conditions, we show the existence of families of subsets of $[n]$ that are affectance-selective for a given family of subsets of $[n]$. We also present randomized and deterministic distributed protocols for \diss based on those affectance-selective families, and we provide running time theoretical guarantees. 

\mig{
Our approach combines an engineering solution with theoretical guarantees. That is, we provide a method to characterize the network with a global measure of affectance based on measurements of interference in the specific deployment area. Then, our protocols distributedly produce an ad-hoc transmissions schedule for dissemination. The randomized protocol only requires knowledge of the network characterization (which could be hardwired), whereas the deterministic protocol requires full knowledge of the affectance values and it is computationally intensive. 
Similar approaches have been explored in practice, e.g. Conflict Maps (CMAP)~\cite{cmap}, where nodes probe the network to build a map of conflicting transmissions.
}


\mig{
In order to show the impact of a more accurate model of interference, we run simulations for a real-world deployment area. We compare the performance of our randomized protocol with previous protocols designed for the \RN and SINR models. Our experimental results expose a striking improvement in running time. Notably, this improvement does not come from algorithmic novelty, since all three protocols rely on transmitting with some probability, but from the careful choice of this probability as a function of the interference measured experimentally. 
}

\paragraph{Roadmap}
In Section~\ref{sec:relwork} we overview previous related work. In Section~\ref{sec:model} we specify the details of our models and the relation between affectance-selective families and dissemination in bipartite graphs. In Section~\ref{sec:results} we specify the results obtained highlighting the novelty of our contribution. Section~\ref{sec:analysis} contains our analysis and the protocols presented, and in Section~\ref{section:sim} we present our simulation results.


\section{Related Work}
\label{sec:relwork}

Before our work in~\cite{KMR_fomc14,KowalskiMZarxiv15}, the {\em generalized affectance} model was introduced and used only in the context of one-hop
communication, more specifically, to link scheduling by Kesselheim~\cite{Kaff,KVaff}. 
He also showed how to use it for dynamic link scheduling in batches.
This model was inspired by the affectance parameter introduced in the more restricted SINR setting~\cite{HWaff}.
They give a characteristic of a set of links, based on affectance, that influence
the time of successful scheduling these links under the SINR model.
In~\cite{KMR_fomc14,KowalskiMZarxiv15}, we generalized this characteristic, called the maximum average tree-layer affectance,
to be applicable to multi-hop communication tasks
such as broadcast, together with another characteristic, called the maximum path affectance.

\diss is closely related to the combinatorial problem of computing selective families ad-hoc for a given family of sets. 
The notion of selective families was introduced in~\cite{CGGPRselFam} and it is defined as follows.
Given any set of items $U$, a family $\cF$ of subsets of $U$ is called $k$-selective for the set $U$ if and only if for any $X\subseteq U$, such that $|X|\leq k$, there is a set $Y\in \cF$ satisfying $|X\cap Y|=1$.
Here, we introduce the concept of affectance-selective families, taking into account the specific conditions to achieve a successful transmission under affectance. 

With respect to selective families, our work can be seen as an extension of~\cite{ClementiCMPSapprox01} to affectance. Indeed, in~\cite{ClementiCMPSapprox01}, Clementi et al. showed how to compute selective families ad-hoc for a given family. That is, their algorithm can be used for dissemination under the \RN model. For instance, the input families can be seen as the different subsets of nodes that may be active at a given time, or as in \diss as the subsets of transmitters connected to each receiver. 
Here, we revisit this problem under affectance, that is, we show the existence of affectance-selective families (the precise notion is defined in Section~\ref{sec:model}), we present randomized and deterministic protocols to solve \diss based on the affectance-selective families, and we analyze their performance.



\section{Model and Problem}
\label{sec:model}
We model the network topology as a bipartite graph $G=(V,W,E)$, 
where $V$ is the set of transmitters, 
$W$ is the set of receivers, 
$|V|=|W|=n$, 
and $E$ is the set of links from $V$ to $W$. 
That is, for every $(v,w)\in E$, we have $v\in V$ and $w\in W$.
For each $w\in W$, we denote by $E_w$ the set of links incoming to receiver $w$, and 
by $F_w$ the set of transmitters of those links.

Following~\cite{KowalskiMZarxiv15}, we model the interference among transmissions with an \defn{affectance matrix} 
%
$$A=\bigg[a(u,(v,w))\bigg]_{\substack{u\in V\\(v,w)\in E}} \ ,$$ 
%
where $a(u,(v,w))$ is a real number in $[0,1]$ quantifying the interference that the transmitter $u$ introduces to the communication through link $(v,w)$. 
We denote $a_{V'}((v,w))$ as the total affectance of a set of transmitters $V'\subseteq V$ on a link $(v,w)$ (i.e., the sum of affectances on $(v,w)$ over all nodes in $V'$), 
and further, $a_{V'}(E')$ as the total affectance of a set of transmitters $V'\subseteq V$ on a set of links $E'\subseteq E$ 
(i.e., the sum of affectances of $V'$ over all links in $E'$).
We do not restrict the affectance function, as long as its effect is additive; that is, 
\begin{align*}
a_{V'}((v,w)) &= \sum_{u\in V'}a(u,(v,w))\textrm{ , and }\\
a_{V'}(E') &= \sum_{(v,w)\in E'}a_{V'}((v,w)) \ .
\end{align*}

Without loss of generality we assume that time is slotted. 
Then, under the above affectance model, a \defn{successful transmission} in a time slot $t$ is defined as follows.
For any link $(v,w)\in E$, a transmission from $v$ is received at $w$ in time slot $t$ if and only if: 
\begin{itemize}
\item
$v$ transmits in time slot $t$, and 
\item
$a_{\mathcal{T} (t)}((v,w))< 1$, where $\mathcal{T} (t)\subseteq V$ is the set of nodes transmitting in time slot $t$.
\end{itemize}
The event of a non-successful transmission, that is, when the affectance is at least $1$, is called a \defn{collision}. 
We assume that a node listening to the channel cannot distinguish between a collision and background noise present in the channel in absence of transmissions; in other words,
the model is without collision detection. 

Under the model above, the \defn{\diss} problem is defined as follows: 
for each node $w\in W$, $w$ must receive a successful transmission from some node in $F_w$.

We define affectance-selective families as a purely combinatorial problem on a family of subsets of integers and a matrix of real numbers. 
(Refer to Section~\ref{sec:relwork} for a definition of classic selective families.)
The relation with \diss is the following. 
For each receiver $w\in W$, consider the set $F_w\subseteq V$ of transmitters connected to $w$. These sets of transmitters define a family $\cF$ of subsets of nodes in $V$. 
On the other hand, for a given \diss protocol, the schedule of transmissions from nodes in $V$ can also be viewed as a family $\cS$ of subsets of nodes. Specifically, for each time slot $t$, the subset of nodes in $V$ transmitting in $t$ is a member of $\cS$. 
In the \RN model, the family $\cS$ is called selective on the family $\cF$ if and only if for any $F_w\in\cF$ there is some $S_t\in \cS$ such that $|S_t\cap F_w|=1$. This is because $w$ successfully 
receives a message if and only if exactly one node in $F_w$ transmits. 
Given an integer $n>0$, 
consider a family $\cF = \{F_1,F_2,\dots,F_n\}$ of subsets of integers in $[n]$.
Let $A$ be a matrix of real numbers in $[0,1]$ defined on $\cF$
in such a way that for each $u\in[n]$ there is a corresponding row, and 
for each $w\in[n]$ and each $v\in F_w$ there is a column in $A$ corresponding to the pair $(v,w)$.
%
Then, we say that 
a family $\cS = \{S_1,S_2,\dots,S_s\}$ of subsets of $[n]$ is \defn{affectance-selective} on the family $\cF$ if
for each $w\in[n]$ there exists $j\in[s]$ such that: 
\begin{itemize}
\item
$|F_w\cap S_j|\geq 1$, and 
\item
for some $v\in (F_w\cap S_j)$ it is $\sum_{u\in S_j}a(u,(v,w))<1$.
\end{itemize}
We say that the family $\cS$ has \defn{length} $s$, and that each $w$ is \defn{affectance-selected}, or simply \defn{selected} for short.

In terms of \diss, 
labeling the transmitters as well as the receivers with consecutive integers in $[n]$, 
each $F_w\in \cF$ is the subset of transmitters connected to receiver $w$, $A$ is the affectance matrix, and each value $a(u,(v,w))$ in $A$ corresponds to the affectance of node $u$ on link $(v,w)$. Then, the family $\cS$ is a solution for \diss setting each node in set $S_t\in\cS$ to transmit in time slot $t$, for each $t\in[s]$.


\section{Our Results}
\label{sec:results}

In this work, for a given family $\cF = \{F_1,F_2,\dots,F_n\}$ of subsets of integers in $[n]$ and a given affectance matrix $A$, we first show the existence of a family $\cS$ of subsets of $[n]$ that is affectance-selective on $\cF$. Under certain conditions on the relation between $\cF$ and $A$, the family $\cS$ is proved to have a number of sets that is in $O(1+\log n\log \overline{A})$. That is, at most logarithmic on $n$ and logarithmic on the \defn{maximum average affectance} $\overline{A}$. The latter is a characterization based on $\cF$ and $A$. Specifically, $$\overline{A}=\max_{w\in[n]}\max_{F\subseteq F_w}\sum_{v\in F}\sum_{u\in [n]}a(u,(v,w))/|F|.$$

The proof of that bound is existential because it is based on the probabilistic method (as in~\cite{ClementiCMPSapprox01}). Nevertheless, it provides a method to derive algorithms for \diss. We present two \diss distributed protocols, one randomized and one deterministic. We show that both protocols have the same running time guarantee, which is asymptotically the same as the size of the affectance-selective family shown. That is, $O(1+\log n\log \overline{A})$. The randomized protocol is Monte Carlo, it is very simple (a version of Decay~\cite{decay}), and only requires knowledge of $n$, $\overline{A}$, and two constants. The deterministic protocol (inspired on~\cite{ClementiCMPSapprox01}) provides worst-case guarantees, but nodes need to know the topology and the affectance matrix $A$, and its computational complexity is exponential.

\mig{
We also include simulations to evaluate the impact of using a more accurate model of interference. We compare our randomized protocol with previous work for the \RN and SINR models. Our experimental results show a striking improvement in performance because the \RN protocol neglects interference from non-neighboring nodes, whereas SINR protocols do not take advantage of low interference from nodes that, although located at a short distance, are blocked by obstacles.
Our results also show that for the particular inputs tested our randomized protocol performs better than predicted by our theoretical analysis.
}


\section{Analysis}
\label{sec:analysis}

\subsection{Existence of an Affectance-selective Family of Polylogarithmic Size}

\begin{theorem}
\label{thm:family}
For any $n>0$,
consider a family $\cF = \{F_1,F_2,\dots,F_n\}$ of subsets of integers in $[n]$ 
and any affectance matrix $A$ defined on $\cF$.
For each $w\in[n]$, let  $\overline{A}_w=\max_{F\subseteq F_w}\sum_{v\in F}\sum_{u\in [n]}a(u,(v,w))/|F|$ be the maximum average affectance on $w$. 
If there exists a constant $c>1$ such that $\overline{A}_w \leq c|F_w|$ for all $w\in[n]$, 
then, there exists a family $\cS = \{S_1,S_2,\dots,S_s\}$ 
that is affectance-selective on $\cF$, and its size $s$ satisfies
$$s\in O\left(1+\log n\log \overline{A} \right),$$
where $\overline{A}=\max_{w\in[n]}\overline{A}_w$ is the maximum average affectance.
\end{theorem}

\begin{proof}

We prove the claim using the probabilistic method. That is, we show a randomly generated family $\cS$ such that the probability that $\cS$ does not select some set in $\cF$ is strictly less than one. 

Let $S$ be a subset of $[n]$ defined as follows. For each $v\in [n]$, independently include $v$ in $S$ with some probability $p$ (we will discuss the best choice
for $p$ later). 
Let $X_v$ be a random variable indicating whether $v$ is in $S$ or not.
Let $Z_{w}$ be a random variable indicating whether $w\in [n]$ is selected or not.
The probability that $w$ is not selected given that some $v\in F_w$ is in $S$ is
\begin{align*}
Pr\left(Z_{w}=0\bigg|\sum_{v\in F_w}X_v \geq 1\right) 
&\leq Pr\left(\sum_{u\in [n]} \sum_{v\in F_w} a(u,(v,w)) X_u X_v
\geq \sum_{v\in F_w}X_v\right).
\end{align*}

The above inequality is true because, for $w$ not to be selected, the affectance in all pairs $(v,w)$ such that $v\in F_w$ and $X_v=1$ has to be at least one. 
The inequality is due to the right-hand side possibly including events where some pairs have affectance less than one, but others have affectance larger than one so that the overall sum is still larger than $\sum_{v\in F_w}X_v$.
This right-hand side can be bounded as follows using a Markov-type inequality that can be proved as in~\cite{MitzenmacherUpfal-book2005}. 

\begin{align*}
Pr\left(\sum_{u\in [n]} \sum_{v\in F_w} a(u,(v,w)) X_u X_v
\geq \sum_{v\in F_w}X_v\right)
&\leq E\left(\frac{\sum_{u\in [n]} \sum_{v\in F_w} a(u,(v,w)) X_u X_v}{\sum_{v\in F_w}X_v}\right).
\end{align*}

Replacing  $\overline{A}_w = \max_{F\subseteq F_w}\sum_{v\in F}\sum_{u\in [n]}a(u,(v,w))/|F| \geq \sum_{v\in F'}\sum_{u\in [n]}a(u,(v,w))/|F'|$, for any $F'\subseteq F_w$, we get the following bound.

\begin{align*}
Pr\left(Z_{w}=0\bigg|\sum_{v\in F_w}X_v \geq 1\right) 
&\leq E\left(\frac{\sum_{u\in [n]} \sum_{v\in F_w} a(u,(v,w)) X_u X_v}{\sum_{v\in F_w}X_v}\right)\\
&= E\left(\sum_{u\in [n]} \frac{\sum_{v\in F_w} a(u,(v,w)) X_v}{\sum_{z\in F_w}X_z}X_u \right)\\
&\leq \overline{A}_wp \ .
\end{align*}

Then, we have that
\begin{align}
Pr\left(Z_{w}=0\right) 
&= Pr\left(Z_{w}=0\bigg|\sum_{i\in F_w}X_i \geq 1\right)Pr\left(\sum_{i\in F_w}X_i \geq 1\right)\nonumber\\ 
&+ Pr\left(Z_{w}=0\bigg|\sum_{i\in F_w}X_i = 0\right)Pr\left(\sum_{i\in F_w}X_i =0\right)\nonumber\\
&= Pr\left(Z_{w}=0\bigg|\sum_{i\in F_w}X_i \geq 1\right)(1-(1-p)^{|F_w|}) + (1-p)^{|F_w|}\nonumber\\
&\leq \overline{A}_w p(1-(1-p)^{|F_w|}) + (1-p)^{|F_w|}\label{prob}\\
&= \overline{A}_w p + ( 1 - \overline{A}_w p)(1-p)^{|F_w|}.\nonumber
\end{align}

Consider now a family $\mathcal{S}=\{S_i\}$ of subsets of $[n]$ where $S_i$ is obtained  
including each $v\in [n]$ independently with probability $p=1/b^i$ for $i=0,1,2,\dots,\max\{\lceil\log_b(2\overline{A})\rceil,0\}$ and $b=1+1/(2c)$. 
If $\overline{A}_w\leq 1/(2b)$, replacing in Equation~\ref{prob} we have that $Pr\left(Z_{w}=0\right)\leq 1/(2b)$ for $p=1$, which is strictly smaller than $1$.
Otherwise, if $\overline{A}_w> 1/(2b)$, we know that, for some $i$, it is $1/(2b\overline{A}_w) < p \leq 1/(2\overline{A}_w)$. Replacing, 
\begin{align*}
Pr\left(Z_{w}=0\right) 
&\leq \frac{1}{2}+\left(1-\frac{1}{2b}\right)\left(1-\frac{1}{2b\overline{A}_w}\right)^{|F_w|}.
\end{align*}

Using that $\overline{A}_w \leq c|F_w|$ for some constant $c>1$, we obtain
\begin{align*}
Pr\left(Z_{w}=0\right) 
&\leq \frac{1}{2}+\left(1-\frac{1}{2b}\right)\left(1-\frac{1}{2bc|F_w|}\right)^{|F_w|}, \textrm{ using that $2bc|F_w|>1$,}\\
&\leq \frac{1}{2}+\left(1-\frac{1}{2b}\right)\left(\frac{1}{e}\right)^{1/(2bc)}.
\end{align*}

Replacing $c=1/(2(b-1))$ we get
\begin{align*}
Pr\left(Z_{w}=0\right) 
&\leq \frac{1}{2}+\left(1-\frac{1}{2b}\right)\left(\frac{1}{e}\right)^{(b-1)/b}.
\end{align*}

To show that there is a positive probability that $w$ is selected, we show that for each constant $c$ there is a constant $b=1+1/(2c)$ such that the latter is strictly smaller than $1$ as follows.
\begin{align*}
\frac{1}{2}+\left(1-\frac{1}{2b}\right)\left(\frac{1}{e}\right)^{(b-1)/b} &< 1\\
\left(1-\frac{1}{2b}\right)\left(\frac{1}{e}\right)^{(b-1)/b} &< \frac{1}{2}\\
1-\frac{1}{2b} &< \frac{1}{2}e^{(b-1)/b}\\
1-\frac{1}{2}e^{(b-1)/b} &< \frac{1}{2b}\\
2b- b e^{(b-1)/b} &< 1.
\end{align*}

The left hand side is equal to $1$ for $b=1$ and monotonically decreasing for any $b$ such that $1<b< 1.5$, which is the range of $b=1+1/(2c)$ for any $c>1$.

Having proved that that there is a positive probability that $w$ is selected, we add a multiplicity $m$ on the sets $S_i$ to show that the probability that \emph{some} $w\in[n]$ is not selected is small, as follows.

We redefine $\mathcal{S}$ as the family $\{S_{i,j}\}$ of subsets of $[n]$ where
the set $S_{i,j}$ is obtained including each $v\in [n]$ in $S_{i,j}$ independently with probability $p=1/b^i$, 
for each $i = 0,1,2, \dots, \max\{\lceil\log_b(2\overline{A})\rceil,0\}$ and each $j=1,2,\dots,m$. 

%
Then, the probability that a given $w$ is not selected is
$Pr\left(Z_{w}=0\right) \leq d^m$, where $d<1$ is some constant as shown above.
Using the union bound, the probability that \emph{some} $w\in [n]$ is not selected is
$Pr\left(\exists w\in [n] : Z_{w}=0\right) 
\leq nd^m$,
which is smaller than $1$ for some $m\in \Theta(\log n)$, 
showing the existence of an affectance-selective family $\cS$ of size $O(1+\log n\log\overline{A})$. 
\end{proof}

The bound shown matches the $O(1+\log\Delta\log |\mathcal{F}|)$ bound for the \RN model in~\cite{ClementiCMPSapprox01}, because in our setting the number of subsets to select is $|\mathcal{F}|=n$, and in the \RN model it is $\overline{A} = \Delta-1$.

\subsection{Randomized \diss Protocol}
The proof of Theorem~\ref{thm:family}, showing the existence of an affectance-selective family, yields a Monte Carlo distributed randomized protocol for \diss applicable to settings where the conditions of the theorem hold. I.e., there exists a constant $c$ bounding $ \overline{A}_w \leq c|F_w|$ for each receiver $w$. The protocol requires that all transmitters have knowledge of the maximum average affectance $\overline{A}$, the constant $c$, the number of transmitters $n$, and the constant $d<1$ computed in the proof of Theorem~\ref{thm:family}. The protocol, detailed in Algorithm~\ref{randalg}, is a version of the Decay protocol~\cite{decay} extended to the affectance model. 
Its correctness and running time are established in the following theorem.


\begin{algorithm}[htbp]
\caption{Randomized \diss protocol for each node $v\in V$. $\overline{A}=\max_{w\in W}\overline{A}_w$, is the maximum average affectance, where $\overline{A}_w=\max_{F\subseteq F_w}\sum_{v\in F}\sum_{u\in V}a(u,(v,w))/|F|$ is the maximum average affectance on $w$, $d<1$ is a constant as computed in the proof of Theorem~\ref{thm:family}, and $c>1$ is the constant bounding $ \overline{A}_w \leq c|F_w|$ for each receiver $w$.}
\label{randalg}
\DontPrintSemicolon
	$b\leftarrow 1+1/(2c)$\;
	$m \leftarrow \lceil 2 \log_{1/d} n \rceil$\;
	\For{$i=0,1,2,\dots,\max\{\lceil\log_b(2\overline{A})\rceil,0\}$}{\label{randalgloop}
		\For{$m$ times}{
			transmit with probability $1/b^i$\;
		}
	}
\end{algorithm}

\begin{theorem}
\label{thm:randalg}
Consider a layer of a \RN with affectance matrix $A$ and topology $G=(V,W,E)$, 
where $|V|=|W|=n$, where for each receiver $w\in W$ there is at least one transmitter $v\in V$ such that $(v,w)\in E$.
Then, 
if there exists a constant $c>1$ such that $\overline{A}_w \leq c|F_w|$ for all $w\in W$, 
where $\overline{A}_w=\max_{F\subseteq F_w}\sum_{v\in F}\sum_{u\in V}a(u,(v,w))/|F|$ is the maximum average affectance on $w$,
Algorithm~\ref{randalg} solves the \diss problem with high probability~\footnote{We say that an event occurs \emph{with high probability} if it occurs with probability at least $1-1/n^\kappa$, for some constant $\kappa>0$.}, and the running time is in $O(1+\log n\log \overline{A})$,
where $\overline{A}=\max_{w\in W}\overline{A}_w$ is the maximum average affectance.
\end{theorem}
\begin{proof}
The first claim follows from the proof of Theorem~\ref{thm:family}, together with computing the value $m$ that makes 
$Pr\left(\exists w\in [n] : Z_{w}=0\right) 
\leq nd^m \leq 1/n$.
The running time follows from the number of iterations in Algorithm~\ref{randalg}.
\end{proof}

For settings where only $n$ and $c$ are known to the transmitters, we can run the loop in Line~\ref{randalgloop} of Algorithm~\ref{randalg} for $\lceil\log_b(2(n-1))\rceil$ times, since we know that $\overline{A}_w\leq (n-1)$ for any $w\in W$. The running time in that case would be $1+O(\log^2 n)$ steps.

\subsection{Deterministic \diss Protocol}



Algorithm~\ref{randalg} is simple and it is easily distributed because only requires knowledge of a few global parameters (namely $\overline{A}$, $c$, and $n$), and also does not require intensive computations at each node. However, the running time guarantee is only stochastic. In this section we present a deterministic algorithm that provides the same running time guarantee but worst-case, although to implement it distributedly knowledge of the graph $G$ and the affectance matrix $A$ is required.  

The ideas of algorithm $greedy_{MSF(\Delta)}$~\cite{ClementiCMPSapprox01} can be re-used here to compute a transmission schedule that solves \diss, but $greedy_{MSF(\Delta)}$ cannot be used as-is because it does not cope with affectance or families of sets with different sizes.
So, building upon the ideas of $greedy_{MSF(\Delta)}$, we present in this section an algorithm for \diss under the affectance model. 
That is, the transmission schedule is computed to cope with affectance, and without assuming anything about the number of neighbors of each receiver. 
We specify such protocol in Algorithm~\ref{detalg} and an explanation of the details follow.

\begin{algorithm}[htbp]
\caption{Deterministic \diss protocol for each node $v\in V$. $\overline{A}=\max_{w\in W}\overline{A}_w$, is the maximum average affectance, where $\overline{A}_w=\max_{F\subseteq F_w}\sum_{v\in F}\sum_{u\in V}a(u,(v,w))/|F|$ is the maximum average affectance on $w$, and $c>1$ is the constant bounding $ \overline{A}_w \leq c|F_w|$ for each $w$.}
\label{detalg}
\DontPrintSemicolon
	\tcp{Initialization}
	$p \leftarrow 0$\label{probinit}\;
	$b \leftarrow 1+1/(2c)$\;
	$m\leftarrow\max\{\lceil\log_b (2\overline{A})\rceil,0\}$\;
	$W'_0 \leftarrow \{w\in W:\overline{A}_w\leq 1/2\}$\;
	\lFor{$r=1,\dots,m$}{
		$W'_r \leftarrow \{w\in W:b^{r-1}/2<\overline{A}_w\leq b^r/2\}$
	}
	\BlankLine
	\tcp{Protocol}
	\For{each time slot {\bf while} $\exists r=0,1,\dots,m:W'_r\neq \emptyset$}{\label{outerloop}
		\If{$p\leq1/(2b\overline{A})$}{
			$p \leftarrow 1$\label{probone}\;
			$r\leftarrow 0$\;
		}
		set $V'[1\dots n]$ array of booleans\tcp*{$V'[i] \equiv i$ transmits}
		\For{$i=1,2,\dots,n$}{\label{innerloop}
			$\mathop{\mathbb{E}}_{true} \leftarrow \mathop{\mathbb{E}}_{V'[i+1\dots n]}\left(\textrm{\# selected in } W'_r \big| V'[i]=true\right)$\label{exptrue}\;
			$\mathop{\mathbb{E}}_{false} \leftarrow \mathop{\mathbb{E}}_{V'[i+1\dots n]}\left(\textrm{\# selected in } W'_r \big| V'[i]=false\right)$\label{expfalse}\;
			$V'[i]\leftarrow\mathop{\mathbb{E}}_{true} > \mathop{\mathbb{E}}_{false}$\label{endinnerloop}\;
		}
		\lIf{$V'[v]$}{transmit}
		$W'_r\leftarrow W'_r \setminus \left\{w \big| w \textrm{ was selected}\right\}$\;
		$p \leftarrow p/b$\label{probupdate}\;
		$r\leftarrow r+1$\label{roundupdate}\;
	}
\end{algorithm}

The receivers pending to be selected (initially all) are partitioned in subsets so that, for each receiver $w$, it is
\begin{displaymath}
w\in \left\{ \begin{array}{ll}
 W'_0 & \textrm{if $\overline{A}_w\leq 1/(2b)$}\\
 W'_r & \textrm{if $b^{r-1}/2<\overline{A}_w\leq b^r/2$, for $r=0,1,\dots,m$.}
  \end{array} \right.
\end{displaymath}

The expectations in Lines~\ref{exptrue} and~\ref{expfalse} of the protocol correspond to the following.
Recall that we assume the transmitters to be labeled by consecutive integers. That is, the set of transmitters is $V=\{1,2,\dots, n\}$.
Then, in Algorithm~\ref{detalg}, for each time slot $t$, we keep track of whether each node in $V$ transmits or not in an array of booleans $V'$, where index $i$ of the array is true if $i$ transmits in $t$ and false otherwise. The array is filled incrementally for $i=1,2,\dots,n$ as follows. For each index $i$, let $V_{>i}=\{i+1,\dots,n\}$ if $i<n$, or $V_{>i}=\emptyset$ otherwise. Likewise, let $V_{<i}=\{1,\dots,i-1\}$ if $i>1$, or $V_{<i}=\emptyset$ otherwise. 

Then, for each value of $r=0,1,\dots$, taking into account the action of transmitters in $V_{<i}$ that was already decided, we decide whether transmitter $i$ transmits or not in $t$ computing the expected number of receivers from a given subset that will be affectance-selected, if $i$ transmits and the actions of transmitters in $V_{>i}$ is chosen at random with probability $b^{-r}$ (Line~\ref{exptrue}). We do the same for the case that transmitter $i$ does not transmit (Line~\ref{expfalse}). The expectations are taken over the random choice of transmitters in $V_{>i}$. 
Such computation is feasible given that every transmitter $v\in V$ is assumed to know $G=(V,W,E)$ and the affectance matrix $A$. The specific computation of expectations is the following.

The calculation corresponds to the $i$th iteration of the inner loop (Line~\ref{innerloop}) and probability $p=b^{-r}$ for some $r$.
Let $X_{v,i}$ be an indicator variable defined as follows. 
The variable $X_{v,i}$ is random if $v\in V_{>i}$, and deterministic otherwise. 
For each $v\in V_{<i}$, $X_{v,i}=1$ if and only if $V'[v]=true$.
For each $v\in V_{>i}$, $X_{v,i}=1$ with probability $p$ or $X_{v,i}=0$ with probability $1-p$.
Finally, it is $X_{i,i}=1$ to compute the expectation $\mathop{\mathbb{E}}_{true}$ (Line~\ref{exptrue}) or $X_{i,i}=0$ to compute the expectation $\mathop{\mathbb{E}}_{false}$ (Line~\ref{expfalse}).
Also, let $Z_{w,i}$ be a random variable indicating whether receiver $w$ is selected or not. 

Then, it is
\begin{align*}
\mathbb{E}_{V'[i+1\dots n]}\left(\textrm{\# selected in } W'_r \big| V'[i]=true\right)
&= \sum_{w\in W'_r} Z_{w,i} Pr(Z_{w,i}=1|X_{i,i}=1) \\
\mathbb{E}_{V'[i+1\dots n]}\left(\textrm{\# selected in } W'_r \big| V'[i]=false\right)
&= \sum_{w\in W'_r} Z_{w,i} Pr(Z_{w,i}=1|X_{i,i}=0).
\end{align*}
Where,
\begin{align*}
Pr(Z_{w,i}=1) &= Pr\left(\sum_{v\in F_w}X_{v,i}\geq 1 \textrm{ and } \exists v\in F_w:\sum_{u\in V} \sum_{v\in F_w} a(u,(v,w)) X_{u,i} X_{v,i}<1\right).
\end{align*}

In the following theorem, we prove that 
each time the probability $p$ is updated to $1$ (Line~\ref{probone}), 
at least a constant fraction of receivers is selected, 
solving \diss in a logarithmic number of steps.

\begin{theorem}
\label{thm:det}
Consider a layer of a \RN with affectance matrix $A$ and topology $G=(V,W,E)$, 
where $|V|=|W|=n$, where for each receiver $w\in W$ there is at least one transmitter $v\in V$ such that $(v,w)\in E$.
Then, 
if there exists a constant $c>1$ such that $\overline{A}_w \leq c|F_w|$ for all $w\in W$, 
where $\overline{A}_w=\max_{F\subseteq F_w}\sum_{v\in F}\sum_{u\in V}a(u,(v,w))/|F|$ is the maximum average affectance on $w$,
Algorithm~\ref{detalg} solves the \diss problem, and the running time is in $O(1+\log n\log \overline{A})$,
where $\overline{A}=\max_{w\in W}\overline{A}_w$ is the maximum average affectance.
\end{theorem}

\begin{proof}
Algorithm~\ref{detalg} is correct as long as it terminates, as it does not stop until $W'=\emptyset$ (Line~\ref{outerloop}). Then, to prove the claim, it is enough to prove the upper bound on the running time, which we do as follows. 


Consider the execution divided in stages, where a new stage starts each time that $p$ is set to $1$ (Line~\ref{probinit} and Line~\ref{probone}). 
Moreover, consider each stage divided in rounds according to the value of $r$.
That is, starting from round $r=0$ when $p=1$, a new round starts each time that $p$ and $r$ are updated in Lines~\ref{probupdate} and~\ref{roundupdate}.
Thus, each stage is composed by rounds $0,1,2,\dots,m$ when $p=1,b^{-1},b^{-2},\dots,b^{-m}$ respectively, and when $p$ becomes smaller or equal than $1/(2b\overline{A})$, a new stage begins and $p$ is reset to $1$ in Line~\ref{probone}.

We show now that, in any given round $r$, a constant fraction of receivers in $W'_r$ is selected. Thus, a constant fraction of receivers is selected in each stage, which yields $O(\log n)$ stages, each of $O(\log \overline{A})$ rounds, proving the claimed running time. 

Fix any given round $r$ when $p=b^{-r}$. We focus then on showing that a constant fraction of receivers in $W'_r$ is selected, knowing that, 
for each receiver $w\in W'_r$,
if $r=0$ it is 
$\overline{A}_w\leq 1/2$,
and if $r>0$ it is 
$b^{r-1}/2<\overline{A}_w\leq b^r/2$.

We showed in the proof of Theorem~\ref{thm:family} that, for any $w\in [n]$, 
if a subset $S\subseteq [n]$ is chosen including each $v\in [n]$ with a probability $b^{-i}$,
for $i$ such that $1/(2b\overline{A}_w) < b^{-i} \leq 1/(2\overline{A}_w)$,
the probability of selecting $w$ with $S$ is a positive constant $q$. 
The specific bound on $q$ is dependent on whether $\overline{A}_w\leq 1/(2b)$ or not, but still a constant for both cases. 
This bound applies to round $r$ for any receiver $w\in W'_r$ and $S$ a subset of transmitters, each chosen with probability $b^{-r}$.
Thus, the expected number of receivers selected by $S$ from $W'_r$ would be $qW'_r$, that is, a constant fraction $q$.
Let this expectation be denoted as $\mathbb{E}_{X[1\dots n]}(\textrm{\# selected in } W'_r)$, where each $X[i]$ indicates whether $i\in S$.

Then, to complete the proof, now we show that the expected number of receivers selected from $W'_r$ by the set of transmitters defined by the array $V'$ after completing the loop in Lines~\ref{innerloop}-\ref{endinnerloop} (which indeed is the actual number because no random choice is made in the last iteration) is at least $\mathbb{E}_{X[1\dots n]}(\textrm{\# selected in } W'_r)$.
Indeed, we prove the stronger claim that $\max\{\mathbb{E}_{true},\mathbb{E}_{false}\}\geq\mathbb{E}_{X[1\dots n]}(\textrm{\# selected in } W'_r)$ for each iteration of the loop, which we show by induction on the iteration index $i=1,2,\dots,n$. 
For clarity, we denote $\mathbb{E}_{\bullet}(\textrm{\# selected in } W'_r)$ as $\mathbb{E}_{\bullet}(\textrm{\#})$.
For $i=1$, we have that 
\begin{align*}
\mathbb{E}_{true} &= \mathbb{E}_{V'[2\dots n]}\left(\textrm{\#} \big| V'[1]=true\right) = \mathbb{E}_{X[2\dots n]}\left(\textrm{\#} \big| X[1]=true\right),\\
\mathbb{E}_{false} &= \mathbb{E}_{V'[2\dots n]}\left(\textrm{\#} \big| V'[1]=false\right) = \mathbb{E}_{X[2\dots n]}\left(\textrm{\#} \big| X[1]=false\right).
\end{align*}

Given that $\mathbb{E}_{X[1\dots n]}\left(\textrm{\#}\right) = p\mathbb{E}_{X[2\dots n]}\left(\textrm{\#} \big| X[1]=true\right) + (1-p) \mathbb{E}_{X[2\dots n]}\left(\textrm{\#} \big| X[1]=false\right)$, the claim is true.
Now, assuming that the claim is true for iteration $i-1$, we want to prove that $\max\{\mathbb{E}_{true},\mathbb{E}_{false}\}\geq\mathbb{E}_{X[1\dots n]}\left(\textrm{\#}\right)$ for iteration $i$, where
\begin{align*}
\mathbb{E}_{true} 
&= \mathbb{E}_{V'[i+1\dots n]}\left(\textrm{\#}\big| V'[i]=true\right) \\
\mathbb{E}_{false} 
&= \mathbb{E}_{V'[i+1\dots n]}\left(\textrm{\#}\big| V'[i]=false\right).
\end{align*}

%
%
%

By inductive hypothesis we know that
\begin{align}
\max\{ \mathbb{E}_{V'[i\dots n]}\left(\textrm{\#}\big| V'[i-1]=true\right) 
,\mathbb{E}_{V'[i\dots n]}\left(\textrm{\#}\big| V'[i-1]=false\right)\}
&\geq \mathbb{E}_{X[1\dots n]}\left(\textrm{\#}\right).\label{eqfhi}
\end{align}

Call $\mathbb{E}_{V'[i\dots n]}\left(\textrm{\#}\right)$ the expected number of receivers selected after we fix the value of $V'[i-1]$ in Line~\ref{endinnerloop}. That is,
\begin{align*}
\mathbb{E}_{V'[i\dots n]}\left(\textrm{\#}\right) 
&= \max\{ \mathbb{E}_{V'[i\dots n]}\left(\textrm{\#}\big| V'[i-1]=true\right) 
,\mathbb{E}_{V'[i\dots n]}\left(\textrm{\#}\big| V'[i-1]=false\right)\}.
\end{align*}

Replacing in Equation~\ref{eqfhi}, we have that
\begin{align}
\mathbb{E}_{V'[i\dots n]}\left(\textrm{\#}\right) 
&\geq \mathbb{E}_{X[1\dots n]}\left(\textrm{\#}\right).\label{eqhi}
\end{align}

We also have that
\begin{align}
\mathbb{E}_{V'[i\dots n]}\left(\textrm{\#}\right) 
&= p \mathbb{E}_{V'[i+1\dots n]}\left(\textrm{\#}\big| V'[i]=true\right) 
+ (1-p) \mathbb{E}_{V'[i+1\dots n]}\left(\textrm{\#}\big| V'[i]=false\right)\nonumber\\ 
&\leq \max\{ \mathbb{E}_{V'[i+1\dots n]}\left(\textrm{\#}\big| V'[i]=true\right) 
,\mathbb{E}_{V'[i+1\dots n]}\left(\textrm{\#}\big| V'[i]=false\right)\}.\label{eqmax}
\end{align}

Combining inequalities~\ref{eqmax} and~\ref{eqhi}, the claim follows.

\end{proof}


\section{Simulations}
\label{section:sim}

In this section we present our simulations, developed to evaluate the impact of a more accurate model of interference on \diss. For that purpose, we run simulations for a real-world deployment area, comparing the performance of our randomized protocol with previous protocols designed for the \RN and SINR models. The details follow.

We used as a model of a network deployment area the floor plan of the Seidenberg School of Computer Science and Information Systems at Pace University,  considering nodes installed in the intersections of each square of four ceiling panels (see Figure~\ref{fig:floorplan}). To evaluate \diss, we focused on one layer of this network going across various offices (see Figure~\ref{fig:layer}). For simplicity, to evaluate performance as $n$ grows, we replicated the same office multiple times in a layer.  

\begin{figure}
\centering
\subfloat[Seidenberg School of CSIS floor plan\label{fig:floorplan}]{
\includegraphics[width=0.47\textwidth]{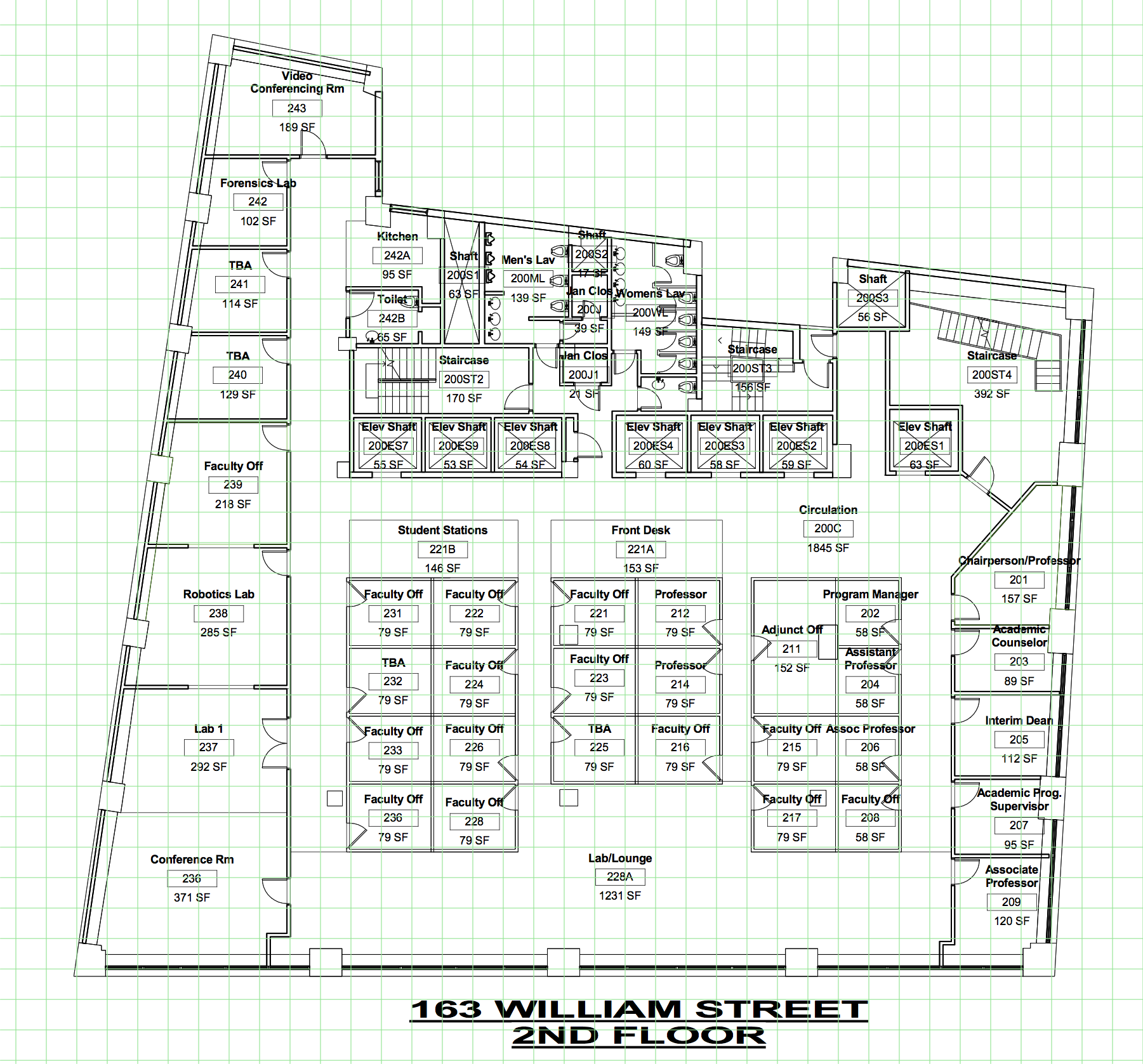}
}
~
\subfloat[A layer of the network grid\label{fig:layer}]{
\includegraphics[width=0.47\textwidth]{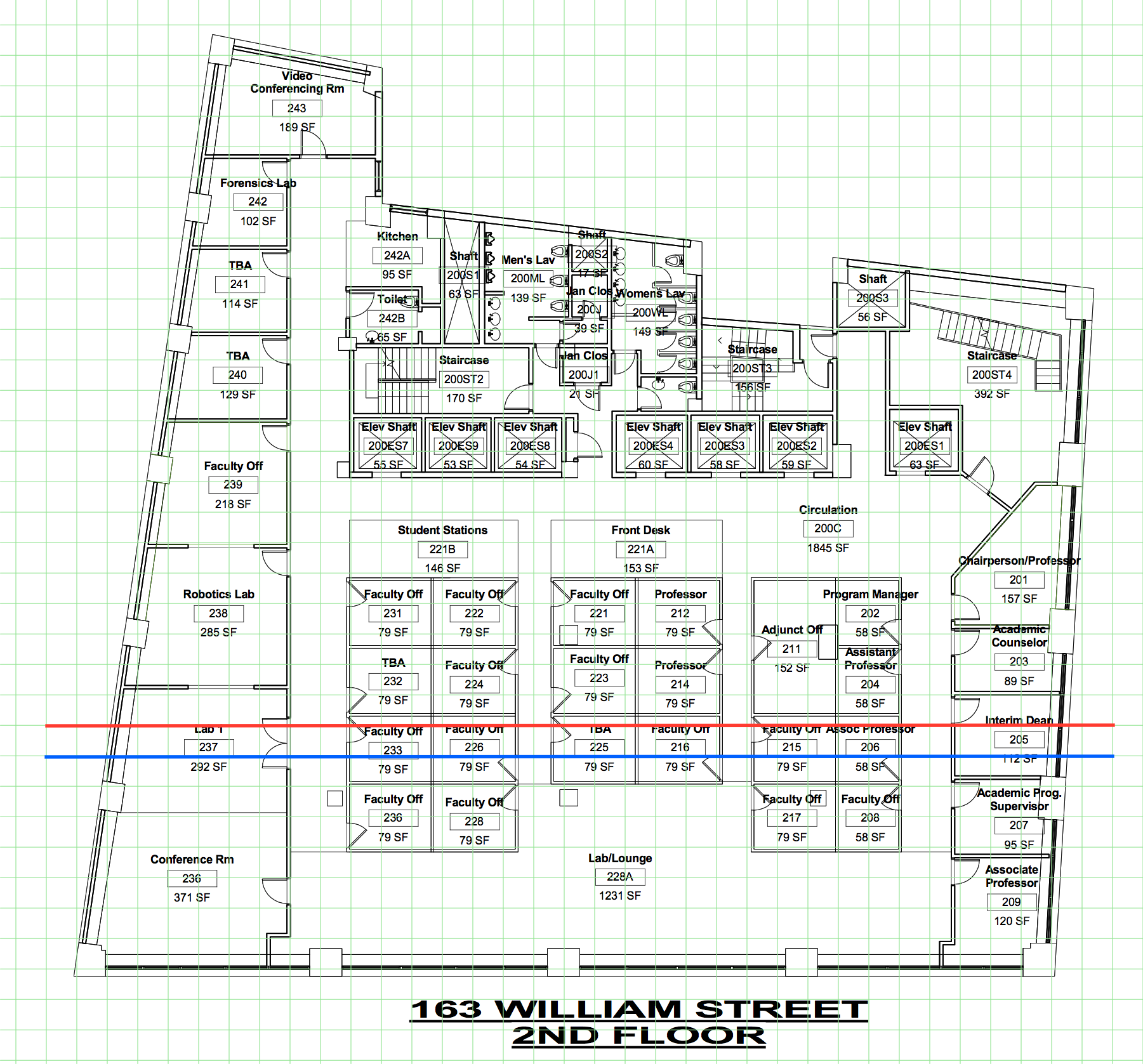}
}
\par
\vspace{0.2in}
\subfloat[The network topology\label{fig:network}]{
\includegraphics[width=0.47\textwidth]{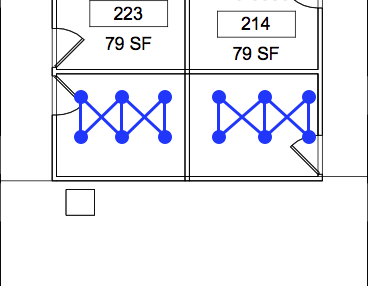}
}
\subfloat[Example of affectance\label{fig:affectance}]{
\includegraphics[width=0.47\textwidth]{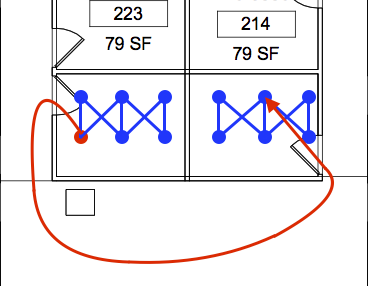}
}
\caption{Illustration of network deployment.}
\end{figure}

The walls of these offices have a metallic structure. Hence, each office behaves as a Faraday cage blocking radio transmissions (specially millimeter wave). Consequently, most of the radio waves propagate through doors (which are not metallic). We fixed the radio transmission power to be large enough to reach five grid cells, so that transmissions from layer to layer are possible. So, given the offices dimensions, transmitters within an office are connected to all receivers. On the other hand, the interference to other offices in the same layer is approximated by adding ten grid cells for each office of distance. The resulting topology can be seen in Figure~\ref{fig:network}, whereas the reason why affectance is more accurate than interference based on Euclidean distance is illustrated in Figure~\ref{fig:affectance}. For instance, it can be seen that transmitters that are close to a wall in one office have low affectance on links that are close to other side of that wall in the contiguous office, even though they are separated by only one grid-cell in Euclidean distance.

Using the network topology and the resulting affectance matrix described above as input, and for $n=6,9,12,\dots,42$, we simulated our randomized protocol in Algorithm~\ref{randalg}, which requires knowledge of only global variables $n$, $c$, and $\overline{A}$. (Refer to Algorithm~\ref{randalg} for further details.) 
For comparison, we also simulated protocols designed for the \RN and SINR models of interference on the same inputs, but considering a transmission successful under the affectance model constraints, as defined in Section~\ref{sec:model}.
We did not simulate our deterministic protocol in Algorithm~\ref{detalg} because the schedule computation has exponential complexity.

For the \RN model of interference, we simulated the classic Decay~\cite{decay} protocol, whereas for SINR we simulated the Broadcast protocol in Algorithm 1 in~\cite{JurdzinskiKRS13}. 
(Most of the work for SINR is oriented to link scheduling, which cannot be accurately mapped to \diss or Broadcast.) 
The former requires knowledge of global variable $\Delta$, which is the maximum in-degree in the network, whereas the latter requires knowledge of global variables $density$ and $dilution$, as defined in~\cite{JurdzinskiKRS13}. 
All three protocols provide guarantees on the number of rounds of communication needed to complete Broadcast, but running them for that fixed time would not provide any performance comparison. Instead, for each of the protocols we measured the number of rounds of communication passed until all receivers have received the message. 
In Algorithms~\ref{algRN} and~\ref{algSINR} we specify how we adapted the \RN and SINR protocols respectively for our simulations.
The results of the simulations are plotted in Figure~\ref{plot} and analyzed in the following section.

\begin{algorithm}
\caption{Decay protocol~\cite{decay} for each transmitter $v\in V$. $\Delta$ is the maximum in-degree of the network.}
\label{algRN}
\DontPrintSemicolon
	$rounds \leftarrow 0$\;
	$counter \leftarrow 0$\;
	\While{$\exists w\in W : w$ did not receive}{
		$rounds ++$\;
		\lIf{$counter = 0$}{$transmit \leftarrow true$}
		\If{$transmit = true$}{
			$v$ transmits the message\;
			with probability $1/2$ set $transmit \leftarrow false$\;
		}
		$counter ++$\;
		\lIf{$counter=2\lceil \log \Delta\rceil$}{$counter \leftarrow 0$}
	}
	\Return{$rounds$}\;
\end{algorithm}

\begin{algorithm}
\caption{Algorithm 1 in~\cite{JurdzinskiKRS13} for each transmitter $v\in V$. $density$ and $dilution$ are parameters of the network as defined in~\cite{JurdzinskiKRS13}.}
\label{algSINR}
\DontPrintSemicolon
	$rounds \leftarrow 0$\;
	\While{$\exists w\in W : w$ did not receive}{
		$rounds ++$\;
		\If{$rounds \equiv v \mod dilution$}{
			with probability $1/density$, $v$ transmits the message\;
		}
	}
	\Return{$rounds$}\;
\end{algorithm}

\section{Conclusions}

\begin{figure}
\centering
	\includegraphics[width=\textwidth]{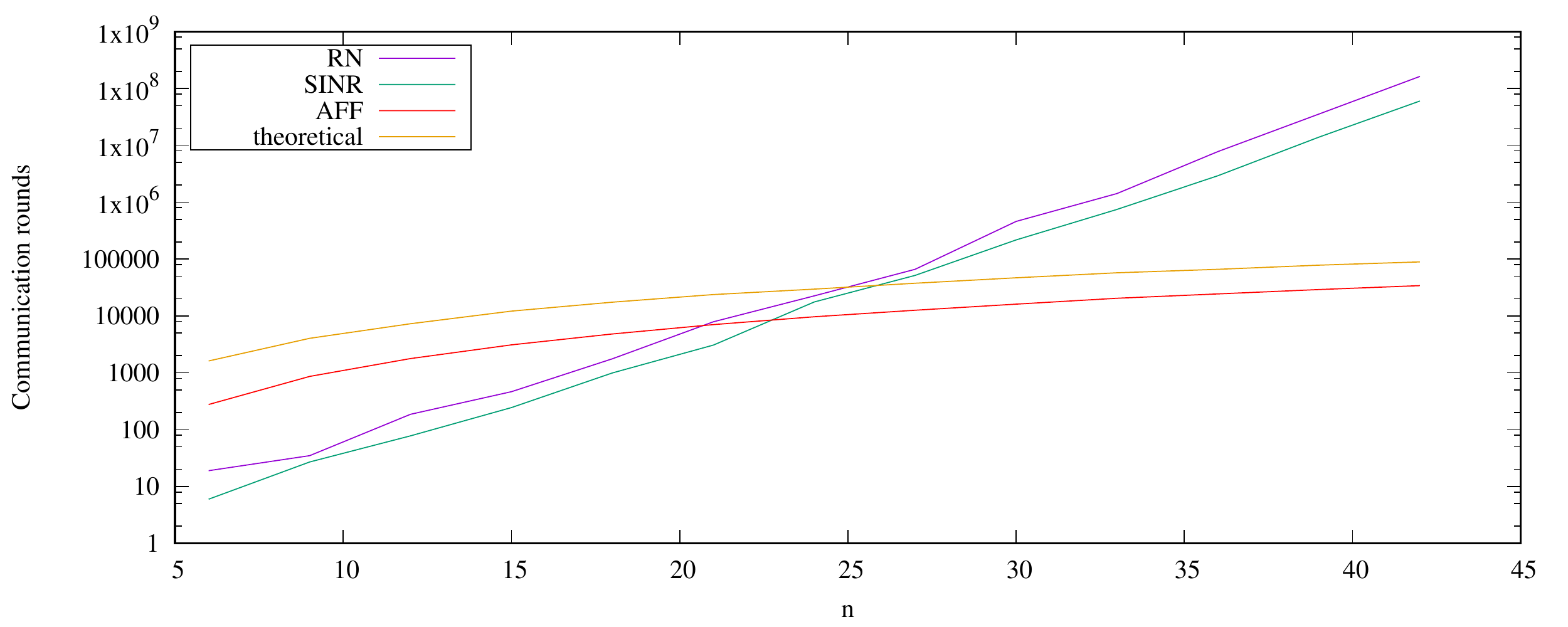}
\caption{Simulation results.}
\label{plot}
\end{figure}

As seen in the plot of Figure~\ref{plot}, our experimental results show a striking improvement in performance of our protocol with respect to Algorithms~\ref{algRN} and~\ref{algSINR}. Indeed, the running times of Algorithms~\ref{algRN} and~\ref{algSINR} grow exponentially with $n$ (the scale of the y axis is logarithmic), whereas our algorithm's running time grows exponentiallly slower. Moreover, the plot shows also the theoretical upper bound proved in Theorem~\ref{thm:randalg}. It can be seen that in these simulations our protocol performs better than the theoretical guarantees. 
This difference in performance could be due to an algorithmic improvement. However, at their core, all three algorithms are based on iteratively choosing to transmit with some probability. Thus, we conclude that the improvement is due to a careful choice of such transmission probability, making it a function of the network characteristic derived from the interference measured experimentally, rather than due to algorithmic novelty. This conclusion should not come as a surprise, given that Algorithm~\ref{algRN} was designed neglecting interference from non-neighboring nodes, whereas Algorithm~\ref{algSINR} does not take advantage of low interference from nodes that, although located at a short distance, are blocked by obstacles. 
Therefore, the results of our experimental evaluation show the importance of studying information dissemination under more accurate models of interference.




\bibliographystyle{plain}
\bibliography{references,references2}



\end{document}